\documentclass{llncs}

\newcommand{\llncs}{}
\input{preamble}
\usepackage{fullpage}

\title{Travelling on Graphs with Small Highway Dimension}
\date{}
\author{Yann~Disser\inst{1}\fnmsep\thanks{Supported by the `Excellence Initiative' of the German Federal and State Governments and the Graduate School~CE at TU~Darmstadt.} \and 
Andreas~Emil~Feldmann\inst{2}\fnmsep\thanks{Supported by the Czech Science 
Foundation GA{\v C}R (grant \#17-10090Y), and by the Center for Foundations of 
Modern Computer Science (Charles Univ.\ project UNCE/SCI/004).}
\and Max~Klimm\inst{3}\fnmsep\thanks{Supported by the German Research Foundation 
(DFG) as part of Math$^+$ (project~AA3-4).}
\and Jochen~K\"onemann\inst{4}\fnmsep\thanks{Supported by the
  Discovery Grant Program of the Natural Sciences and Engineering
  Research Council of Canada} 
}
\institute{TU Darmstadt, Germany, \email{disser@mathematik.tu-darmstadt.de} \and 
Charles University in Prague, Czechia, \email{feldmann.a.e@gmail.com} \and 
Humboldt-Universit\"at zu Berlin, Germany, \email{max.klimm@hu-berlin.de} \and
University of Waterloo, Canada, \email{jochen@uwaterloo.ca}
}

\begin{document}

\tikzset{snode/.style={fill,circle,inner sep=2pt,outer sep=0pt}}
\tikzset{hnode/.style={fill=white,draw=black,inner sep=3pt,outer sep=0pt}}

\maketitle
\setcounter{footnote}{0}

\begin{abstract}
We study the Travelling Salesperson (TSP) and the Steiner Tree problem 
(STP) in graphs of low highway dimension. This graph parameter was introduced by 
Abraham et al.~[SODA~2010] as a model for transportation networks, on which TSP 
and STP naturally occur for various applications in logistics. It was 
previously shown [Feldmann~et al.~ICALP~2015] that these problems admit a 
quasi-polynomial time approximation scheme (QPTAS) on graphs of constant 
highway dimension. We demonstrate that a significant improvement is possible in 
the special case when the highway dimension is~1, for which we present a 
fully-polynomial time approximation scheme (FPTAS). We also prove that STP is 
weakly $\mathsf{NP}$-hard for these restricted graphs. 
For TSP we show $\mathsf{NP}$-hardness for graphs of highway dimension~6, which answers an 
open problem posed in [Feldmann~et al.~ICALP~2015].
\end{abstract}

\section{Introduction}

Two fundamental optimization problems already included in Karp's initial list 
of 21 $\mathsf{NP}$-complete problems~\cite{karp1972} are the \pname{Travelling 
Salesperson} problem~(\pname{TSP}) and the \pname{Steiner Tree} 
problem~(\pname{STP}). Given an undirected graph $G = (V,E)$ with non-negative 
edge weights $w:E\to\mathbb{R}^+$, the \pname{TSP} asks to find the shortest 
closed walk in $G$ visiting all nodes of $V$. Besides its fundamental role in 
computational complexity and combinatorial optimization, this problem has a 
variety of applications ranging from circuit 
manufacturing~\cite{groetschel1991,lenstra1975} and scientific 
imaging~\cite{bland1989} to vehicle routing problems~\cite{laporte1985} in 
transportation networks. For the \pname{STP}, a subset $R \subseteq V$ of nodes 
is marked as \emph{terminals}. The task is to find a weight-minimal connected 
subgraph of $G$ containing the terminals. It has plenty of fundamental 
applications in network design including telecommunication 
networks~\cite{ljubic2006}, computer vision~\cite{chen2017}, circuit 
design~\cite{held2011}, and computational 
biology~\cite{chowdhury2013,loboda2016}, but also lies at the heart of line 
planning in public transportation~\cite{borndorfer2009line}.

Both \pname{TSP} and \pname{STP} are $\mathsf{APX}$-hard in 
general~\cite{chlebik2008,karpinski2015,lampis2014, 
arora1992,bern1989, papadimitriou2006} implying that, unless~$\mathsf{P} = 
\mathsf{NP}$, none of these problems admit a \emph{polynomial-time 
approximation scheme (PTAS)}, i.e., an algorithm that computes a 
$(1+\eps)$-approximation in polynomial time for any given constant $\eps>0$. On 
the other hand, for restricted inputs PTASs do exist, e.g., for planar 
graphs~\cite{borradaile2007Steiner,klein2008TSP,grigni1995,arora1998},
Euclidean and Manhattan 
metrics~\mbox{\cite{arora1998k-median,mitchell1999guillotine}}, 
and more generally low doubling\footnote{A metric is said to have 
\emph{doubling dimension $d$} if for all $r>0$ every ball of radius~$r$ can be 
covered by at most $2^d$ balls of half the radius $r/2$.} 
metrics~\cite{bartal2012traveling}.

We study another class of graphs captured by the notion of \emph{highway 
dimension}, which was proposed by~\citet{abraham2010highway}. This graph 
parameter models transportation networks and is thus of particular importance in 
terms of applications for both \pname{TSP} and \pname{STP}.
On a high level, the highway dimension is based on the empirical observation 
of~\citet{bast2007transit,bast2009ultrafast} that travelling from a 
point in a network to a sufficiently distant point on a shortest path always 
passes through a sparse set of ``hubs''. The following formal definition is 
taken from~\cite{Feldmann15} and follows the lines of 
\citet{abraham2010highway}.\footnote{It is often assumed that all shortest paths 
are unique when defining the highway dimension, since this allows good 
polynomial approximations of this graph parameter~\cite{abraham2011vc}. In this 
work however, we do not rely on these approximations, and thus do not require 
uniqueness of shortest paths.}
Here the \emph{distance} between two vertices is the length of the shortest 
path between them, according to the edge weights. The \emph{ball} $B_v(r)$ of 
radius $r$ around a vertex $v$ contains all 
vertices with distance at most $r$ from $v$.

\begin{dfn}\label{dfn:spc}
For a scale $r \in \mathbb{R}_{>0}$, let $\mc{P}_{(r,2r]}$ denote the set of all 
vertex sets of shortest paths with length in $(r, 2r]$.
A \emph{shortest path cover} for scale $r$ is a hitting set for 
$\mathcal{P}_{(r,2r]}$, i.e., a set $\spc(r) \subseteq V$ such that $|\spc(r) 
\cap P| \neq \emptyset$ for all $P \in \mc{P}_{(r,2r]}$. The vertices of 
$\spc(r)$ are the \emph{hubs} for scale~$r$. 
A shortest path cover $\spc(r)$ is \emph{locally $h$-sparse}, if $|\spc(r) \cap 
B_v(2r)| \leq h$ for all vertices~$v \in V$.
The \emph{highway dimension} of $G$ is the smallest integer $h$ such that there 
is a locally $h$-sparse shortest path cover $\spc(r)$ for every scale 
$r\in\mathbb{R}_{>0}$ in $G$.
\end{dfn}

The algorithmic consequences of this graph parameter were originally studied in 
the context of road 
networks~\cite{abraham2016highway,abraham2010highway,abraham2011vc}, which are 
conjectured to have fairly small highway dimension. Road networks are generally 
non-planar due to overpasses and tunnels, and are also not Euclidean due to 
different driving or transmission speeds. This is even more pronounced in public 
transportation networks, where large stations have many incoming connections and 
plenty of crossing links, making Euclidean (or more generally low doubling) and 
planar metrics unsuitable as models. Here the highway dimension is better 
suited, since longer connections are serviced by larger and sparser stations 
(such as train stations and airports) that can act as hubs.

The main question posed in this paper is whether the structure of graphs with 
low highway dimension admits PTASs for problems such as \pname{TSP} and STP, 
similar to Euclidean or planar instances. It was shown that 
\emph{quasi-polynomial time approximation schemes (QPTASs)} exist for these 
problems~\cite{feldmann2018}, i.e., $(1+\eps)$\hy{}approximation algorithms with 
runtime $2^{\textrm{polylog}(n)}$ assuming that~$\eps$ and the highway dimension 
of the input graph are constants. However it was left open whether this can be 
improved to polynomial time.

\subsection{Our results}

Our main result concerns graphs of the smallest possible highway dimension, and 
shows that for these \emph{fully polynomial time approximation schemes 
(FPTASs)} exist, i.e., a $(1+\eps)$-approximation can be computed in time 
polynomial in both the input size and $1/\eps$. Thus at least for this 
restricted case we obtain a significant improvement over the previously known 
QPTAS~\cite{feldmann2018}.

\begin{thm}\label{thm:main}
Both \pname{Travelling Salesperson} and \pname{Steiner Tree} admit an FPTAS on 
graphs with highway dimension~$1$.
\end{thm}

From an application point of view, so-called hub-and-spoke 
networks that can typically be seen in air traffic networks can be argued to 
have very small highway dimension close to~$1$: their star-like structure 
implies that hubs are needed at the centers of stars only, where all shortest 
paths converge. From a more theoretical viewpoint, we show that surprisingly 
the \pname{STP} problem is non-trivial on graphs 
highway dimension $1$, since it is still $\mathsf{NP}$-hard even on this very
restricted case. 
Interestingly, together with \cref{thm:main} this implies~\cite{vazirani01book} 
that \pname{STP} is \emph{weakly} $\mathsf{NP}$-hard on graphs of highway 
dimension~$1$. This is in contrast to planar graphs or Euclidean metrics, for 
which the problem is strongly $\mathsf{NP}$-hard.

\begin{thm}\label{thm:STP-hard}
The \pname{Steiner Tree} problem is weakly $\mathsf{NP}$-hard on graphs with 
highway dimension $1$.
\end{thm}

It was in fact left as an open problem in~\cite{feldmann2018} to determine the 
hardness of \pname{STP} and also \pname{TSP} on graphs of constant highway 
dimension. \cref{thm:STP-hard} settles this question for STP. We also 
answer the question for TSP, but in this case we are not able to bring down the 
highway dimension to~$1$ so that the following theorem does not complement 
\cref{thm:main} tightly.

\begin{thm}\label{thm:TSP-hard}
The \pname{Travelling Salesperson} problem is $\mathsf{NP}$-hard on graphs with 
highway dimension~$6$.
\end{thm}

\subsection{Techniques}

We present a step towards a better understanding of low highway dimension graphs 
by giving new structural insights on graphs of highway dimension~$1$. It is not 
hard to find examples of (weighted) complete graphs with highway dimension~$1$ 
(cf.~\cite{feldmann2018}), and thus such graphs are not minor-closed. 
Nevertheless, it was suggested in~\cite{feldmann2018} that the \emph{treewidth} 
of low highway dimension graphs might be bounded polylogarithmically in terms of 
the \emph{aspect ratio $\alpha$}, which is the maximum distance divided by the 
minimum distance between any two vertices of the input graph.

\begin{dfn}\label{dfn:treewidth}
A \emph{tree decomposition} of a graph $G=(V,E)$ is a tree $D$ where each node 
$v$ is labelled with a bag $X_v \subseteq V$ of vertices of $G$, such that the 
following holds:
\begin{inparaenum}[(a)]
\item\label{item:tw-union} $\bigcup_{v \in V(D)} X_v = V$, 
\item\label{item:tw-edges} for every edge $\{u,w\}\in E$ there is a
  node $v \in V(D)$ such that $X_v$ contains both $u$ and $w$, and
\item\label{item:tw-vertices} for every $v\in V$ the set $\{u \in V(D)
  \mid v \in X_u\}$ induces a connected subtree of $D$.
\end{inparaenum}
The \emph{width} of the tree decomposition is $\max\{|X_v|-1\mid v
\in V(D)\}$. The \emph{treewidth} of a graph $G$ is the minimum width 
among all tree decompositions for~$G$.
\end{dfn}

As suggested in~\cite{feldmann2018}, one may hope to prove that the treewidth 
of any graph of highway dimension~$h$ is, say, $O(h\,\textrm{polylog}(\alpha))$. 
As argued in \cref{sec:concl}, it unfortunately is unlikely that such a bound 
is generally possible. In contrast to this, our main structural insight on 
graphs of highway dimension $1$ is that they have treewidth $O(\log\alpha)$. 
This implies FPTASs for \pname{TSP} and STP, since we may reduce the aspect 
ratio of any graph with $n$ vertices to $O(n/\eps)$ and then use algorithms by 
\mbox{\citet{bodlaender2013deterministic}} to compute optimum solutions to 
\pname{TSP} and \pname{STP} in graphs of treewidth~$t$ in $2^{O(t)}n$ time. 
Since reducing the aspect ratio distorts the solution by a factor of~$1+\eps$, 
this results in an approximation scheme. Although these are fairly standard 
techniques for metrics (cf.~\cite{feldmann2018}), in our case we need to take 
special care, since we need to bound the treewidth of the graphs resulting from 
this reduction, which the standard techniques do not guarantee.

It remains an intriguing open problem to understand the complexity and 
structure of graphs of constant highway dimension larger than $1$.

\subsection{Related work}

The \pname{Travelling Salesperson} problem (\pname{TSP}) is among Karp's initial 
list of 21 $\mathsf{NP}$-complete problems~\cite{karp1972}. For general metric 
instances, the best known approximation algorithm is due to \citet{Chr76TSP} and 
computes a solution with cost at most $3/2$ times the LP value. For unweighted 
instances, the best known approximation guarantee is $7/5$ and is due to 
\citet{sebo2014}. In general the problem is 
$\mathsf{APX}$-hard~\cite{karpinski2015,lampis2014,papadimitriou2006}. For 
geometric instances where the nodes are points in $\mathbb{R}^d$ and distances 
are given by some $l_p$-norm, there exists a 
PTAS~\cite{arora1998TSP,mitchell1999guillotine} for fixed $d$. When $d= \log n$, 
the problem is $\mathsf{APX}$-hard~\cite{trevisan2000}.  
\citet{krauthgamer2006algorithms} generalized the PTAS to hyperbolic space. 
\citet{grigni1995} gave a PTAS for unweighted planar graphs which was later 
generalized by \mbox{\citet{arora1998}} to the weighted case. For improvements 
of the running time see \citet{klein2008TSP}.

The \pname{Steiner Tree} problem (\pname{STP}) is contained in Karp's list of 
$\mathsf{NP}$-complete problems as well~\cite{karp1972}. The best approximation 
algorithm for general metric instances is due to \citet{byrka2010} and computes 
a solution with cost at most $\ln(4) + \epsilon<1.39$ times that of an LP 
relaxation. 
Their algorithm improved upon previous results by, e.g., 
\citet{robins2005} and \mbox{\citet{hougardy1999}}. 
Also the \pname{STP} is 
$\mathsf{APX}$-hard~\cite{chlebik2008} in general. For Euclidean distances and 
nodes in $\mathbb{R}^d$ with $d$ constant there is a PTAS due to 
\citet{arora1998TSP}. For $d = \log |R| / \log \log |R|$ where $R$ is the 
terminal set, the problem is $\mathsf{APX}$-hard~\cite{trevisan2000}. For planar 
graphs, there is a PTAS for \pname{STP}~\cite{borradaile2007Steiner},
and even for the more general \pname{Steiner Forest} problem for graphs with 
bounded genus~\cite{bateni2011}.
Note that \pname{STP} remains $\mathsf{NP}$-complete for planar 
graphs~\cite{garey1977}.

It is worth mentioning that alternate definitions of the highway dimension 
exist.\footnote{See~\cite[Section~9]{feldmann2018} and \cite{blum2019hierarchy} 
for detailed discussions on different definitions of the highway dimension.} In 
particular, in a follow-up paper to~\cite{abraham2010highway}, 
\citet{abraham2016highway} define a version of the highway dimension, which 
implies that the graphs also have bounded doubling dimension. A related model 
for transportation networks was given by \citet{kosowski2017beyond} via the 
so-called \emph{skeleton dimension}, which also implies bounded doubling 
dimension. Hence for these definitions, \citet{bartal2012traveling} already 
provide a PTAS for \pname{TSP}. The highway dimension definition used here 
(cf.~\cref{dfn:spc}) on the other hand allows for metrics of large doubling 
dimension as noted by \citet{abraham2010highway}: a star has highway 
dimension~$1$ (by using the center vertex to hit all paths), but its doubling 
dimension is unbounded. While it may be reasonable to assume that road networks 
(which are the main concern in the works of 
\citet{abraham2016highway,abraham2010highway,abraham2011vc}) have low doubling 
dimension, there are metrics modelling transportation networks for which it can 
be argued that the doubling dimension is large, while the highway dimension 
should be small. These settings are better captured by \cref{dfn:spc}. For 
instance, the so-called hub-and-spoke networks that can typically be seen in air 
traffic networks are star-like networks and are unlikely to have small doubling 
dimension while still having very small highway dimension close to~$1$. Thus in 
these examples it is reasonable to assume that the doubling dimension is a lot 
larger than the highway dimension. 

\citet{feldmann2018} showed that graphs with low highway dimension can be 
embedded into graphs with low treewidth. This embedding gives rise to a QPTAS 
for both \pname{TSP} and \pname{STP} but also other problems. However, the 
result in~\cite{feldmann2018} is only valid for a less general definition of 
the highway dimension from~\cite{abraham2011vc}, i.e., there are graphs which 
have constant highway dimension according to \cref{dfn:spc} but for which the 
algorithm of~\cite{feldmann2018} cannot be applied. For the less general 
definition from~\cite{abraham2011vc}, \citet{DBLP:conf/esa/BeckerKS18} give a 
PTAS for \pname{Bounded-Capacity Vehicle Routing} in graphs of bounded highway 
dimension. Also the \pname{$k$-Center} problem has been studied on graphs of 
bounded highway dimension, both for the less general 
definition~\cite{DBLP:conf/esa/BeckerKS18} and the more general one used 
here~\cite{DBLP:conf/swat/FeldmannM18,Feldmann15}.

\section{Structure of graphs with highway dimension 1}

In this section, we analyse the structure of graphs with highway dimension $1$. 
To this end, let us fix a graph $G$ with highway dimension $1$ and a shortest 
path cover $\spc(r)$ for each scale $r \in \mathbb{R}^+$. As a preprocessing, 
we remove edges that are longer than the shortest path between their endpoints, 
so that the triangle inequality holds.

We begin by analysing the structure of the graph $G_{\leq 2r}$, which is 
spanned by all edges of the input graph $G$ of length at most $2r$. If $G$ has 
highway dimension~$1$ it exhibits the following key property.

\begin{lem}\label{lem:hd1-G2r}
Let $G$ be a metric graph with highway dimension $1$, $r \in \mathbb{R}^+$ a 
scale, and $\spc(r)$ a shortest path cover for scale $r$. Then, every connected 
component of $G_{\leq 2r}$ contains at most one hub.
\end{lem}
\begin{proof}
For the sake of contradiction, let $r \in \mathbb{R}^+$ and let $x,y\in\spc(r)$ 
be a closest pair of distinct hubs in some component of $G_{\leq 2r}$. Let 
further $P$ be a shortest path in $G_{\leq 2r}$ between $x$ and~$y$ using only 
edges of length at most~$2r$. (Note that $P$ need not be a shortest path between 
$x$ and $y$ in $G$.) In particular, there is no other hub from 
$\spc(r)\setminus\{x,y\}$ along $P$. This implies that every edge of~$P$ that is 
not incident to either $x$ or~$y$ must be of length at most $r$, since otherwise 
the edge would be a shortest path of length $(r,2r]$ between its endpoints 
(using that $G$ is metric) contradicting the fact that $\spc(r)$ is a shortest 
path cover for scale $r$. 

Since the highway dimension of $G$ is $1$, any ball $B_w(2r)$ around a vertex $w 
\in V(P)$ contains at most one of the hubs $x,y \in \spc(r)$. Let $x',y' \in P$ 
be the vertices indicent to $x$ and $y$ along $P$, respectively. Since the 
length of the edge $\{x,x'\}$ is at most~$2r$, the ball~$B_{x'}(2r)$ must 
contain~$x$ and, by the observation above, it cannot contain~$y$ (in particular 
$\{x,y\}$ is not an edge). Symmetrically, the ball $B_{y'}(2r)$ contains~$y$ 
but not~$x$. Consequently, $x' \neq y'$ and neither of these two vertices can be 
a hub of scale~$r$, i.e., the path $P$ contains at least two vertices different 
from $x$ and $y$. 

Let $V_x = \{w \in V : \dist(x,w) < \dist(y,w) \}$ contain all vertices closer 
to $x$ than to $y$, where $\dist(\cdot,\cdot)$ refers to the distance in the 
original graph~$G$. As all edge weights are strictly positive, we have that 
$\dist(x,y) > 0$ and thus~$y \notin V_x$. Since $P$ starts with vertex $x\in 
V_x$ and ends with vertex $y\notin V_x$ we deduce that there is an edge 
$\{u,v\}$ of $P$ such that $u\in V_x$ and $v\notin V_x$. In particular, 
$\dist(x,u)<\dist(y,u)$ and $\dist(y,v)\leq \dist(x,v)$. We must have~$\{u,v\} 
\neq \{y',y\}$, since otherwise $\dist(x,y') < \dist(y,y') \leq 2r$ and hence 
$B_{y'}(2r)$ would contain~$x$. Similarly, we have $\{u,v\} \neq \{x,x'\}$, 
since otherwise $B_{x'}(2r)$ would contain~$y$.
Note that, by definition, $u \neq y$ and $v \neq x$, and hence $x,y \notin 
\{u,v\}$. Consequently, since every edge of~$P$ not incident to either $x$ 
or~$y$ must have length at most $r$, we conclude that $\{u,v\}$ has length at 
most~$r$.

Finally, consider the scale $r' \in \mathbb{R}^+$, defined such that 
$2r'=\dist(x,u)+\dist(u,v)$. Let $Q$ and $Q'$ denote shortest paths between 
$x,u$ and $v,y$ in~$G$, respectively. Then the ball $B_v(2r')$ around $v$ 
contains $Q$ by definition of~$r'$. From $\dist(y,v) \leq \dist(x,v) \leq 
\dist(x,u) + \dist(u,v) = 2r'$ it follows that $B_v(2r')$ contains $Q'$ as well. 
Also, $\dist(y,v) \leq \dist(x,v)$ means that~$B_{v}(2r)$ cannot contain~$x$, 
and hence $2r' = \dist(x,u) + \dist(u,v) \geq \dist(x,v) > 2r$, which 
implies~$r' > r$. W.l.o.g., assume that~$\dist(x,u) \leq \dist(v,y)$ (otherwise 
consider scale $2r' = \dist(y,v) + \dist(u,v)$ and the ball~$B_u(2r')$). Our 
earlier observation that~$\dist(u,v) \leq r$ with~$r < r'$ then 
yields~$\dist(v,y) \geq \dist(x,u) = 2r' - \dist(u,v) > r'$. In other words, the 
lengths of both paths $Q$ and $Q'$ are in~$(r',2r']$, and so they both need to 
contain a hub of~$\spc(r')$. However, by definition of $u,v$, the paths~$Q$ and 
$Q'$ are vertex disjoint, which means that the ball $B_v(2r')$, which contains 
$Q$ and $Q'$, also contains at least two hubs from $\spc(r')$. This is a 
contradiction with~$G$ having highway dimension~$1$.
\qed\end{proof}

Given a graph $G$, we now consider graphs $G_{\leq 2r}$ for exponentially 
growing scales. In particular, for any integer $i\geq 0$ we define the scale 
$r_i=2^i$ and call a connected component of $G_{\leq 2r_i}$ a \emph{level-$i$ 
component}. Note that the level-$i$ components partition the graph $G$, and that 
the level-$i$ components are a \emph{refinement} of the level-$(i+1)$ 
components, i.e., every level-$i$ component is contained in some level-$(i+1)$ 
component. W.l.o.g., we scale the edge weights of the graph such that $\min_{e 
\in E} w(e) = 3$, so that there are no edges on level~$0$, and every level-$0$ 
component is a singleton. Let $\alpha=\frac{\max_{u\neq v}\dist(u,v)} 
{\min_{u\neq v}\dist(u,v)} = \frac{\max_{u\neq v}\dist(u,v)} {3}$ be the aspect 
ratio of~$G$. In our applications we may assume that $G$ is connected, so that 
there is exactly one level-$(1+\lceil\log_2(\alpha)\rceil)$ component containing
all of $G$. 

Since every edge is a shortest path between its endpoints, every edge $e = 
\{u,v\}$ that connects a vertex $u$ of a level-$i$ component $C$ with a vertex 
$v$ outside $C$ is hit by a hub of $\spc(r_j)$, where $j$ is the level for 
which $w(e) \in (r_j,2r_j]$. Moreover, since~$v$ lies outside~$C$, we have $w(e) 
> 2r_i$ and, thus, $j\geq i+1$. 
The following definition captures the set of the hubs through which edges can 
possibly leave $C$. 
\begin{dfn}
	Let $C$ be a level-$i$ component of~$G$. We define the set of 
\emph{interface points} of~$C$ as
	$I_C := \bigcup_{j \geq i} \{u \in \spc(r_j) : \dist_C(u) \leq 
2r_j\}, $
	where $\dist_C(u)$ denotes the minimum distance from $u$ to a vertex in 
$C$ (if $u \in C$, $\dist_C(u) = 0$).
\end{dfn}

Note that, for technical reasons, we explicitly add every hub at level $i$ of a 
component to its set of interface points as well, even if such a hub does not 
connect the component with any vertex outside at distance more than $2r_i$. 

\begin{lem}\label{obs_interface}
If~$G$ has highway dimension~$1$, then each interface~$I_C$ of a level-$i$ 
component~$C$ contains at most one hub for each level $j\geq i$.
\end{lem}
\begin{proof}
Assume that there are two hubs $u,v\in\spc(r_j)$ in $I_C$, and recall that we 
preprocessed the graph so that the triangle inequality holds. Then $u$ and $v$ 
must be contained in the same level-$j$ component~$C'$, since $u$ and $v$ are 
connected to~$C$ with edges of length at most $2r_j$ (or are contained in $C$) 
and $C\subseteq C'$. This contradicts \cref{lem:hd1-G2r}.
\qed\end{proof}

Using level-$i$ components and their interface points we can prove that the 
treewidth of a graph with highway dimension $1$ is bounded in terms of the 
aspect ratio.

\begin{lem}\label{lem:tw}
If a graph $G$ has highway dimension $1$ and aspect ratio $\alpha$, its 
treewidth is at most~$1+\lceil\log_2(\alpha)\rceil$.
\end{lem}
\begin{proof}
The tree decomposition of $G$ is given by the refinement property of level-$i$ 
components. That is, let $D$ be a tree that contains a node $v_C$ for every 
level-$i$ component $C$ for all levels $0\leq i\leq 
1+\lceil\log_2(\alpha)\rceil$. For every node $v_C$ we add an edge in~$D$ to 
node $v_{C'}$, if $C$ is a level-$i$ component, $C'$ is a level-$(i+1)$ 
component, and $C\subseteq C'$. The bag $X_C$ for node $v_C$ contains the 
interface points $I_C$. For a level-$0$ component $C$ the bag $X_C$ additionally 
contains the single vertex $u$ contained in $C$.

Clearly, the tree decomposition has Property~\eqref{item:tw-union} of 
\cref{dfn:treewidth}, since the level\hy{}$0$ components partition the vertices 
of $G$ and every vertex of $G$ is contained in a bag $X_C$ corresponding to a 
level-$0$ component $C$. Also, Property~\eqref{item:tw-edges} is given by the 
bags $X_C$ for level-$0$ components $C$, since for every edge $e$ of $G$ one of 
its endpoints $u$ is a hub of $\spc(r_i)$ where $i$ is such that $w(e) \in 
(r_i,2r_i]$, and the other endpoint $w$ is contained in a level-$0$ component 
$C$, for which~$X_C$ contains $u$ and $w$. 

For Property~\eqref{item:tw-vertices}, first consider a vertex $u$ of $G$, 
which is not contained in any set of interface points for any level-$i$ 
component and any $0\leq i\leq \log_2(\alpha)$. Such a vertex only appears in 
the bag $X_C$ for the level-$0$ component $C$ containing $u$, and thus the node 
$v_C$ for which the bag contains~$u$ trivially induces a connected subtree of 
$D$. 

Any other vertex $u$ of $G$ is an interface point. Let~$i$ be the highest level 
for which $u\in I_C$ for some level-$i$ component $C$. We claim that $u \in C$, 
which implies that $C$ is the unique level-$i$ component containing~$u$ in its 
interface. To show our claim, assume $u \notin C$. Then, by definition, $I_C$ 
contains~$u$ because $u\in\spc(r_j)$ for some $j\geq i$ and $u$ has some 
neighbour at distance at most~$2r_j$ in $C$. Since we preprocessed the graph 
such that every edge is a shortest path between its endpoints, this means that 
there must be an edge~$e = \{u,v\}$ with $w(e) \in (r_j,2r_j]$ and $v \in C$. 
Since $u \notin C$, we have $i < j$. Let~$C'$ be the unique level-$j$ component 
with~$C \subseteq C'$. Then, by definition, $u \in I_{C'}$, which contradicts 
the maximality of~$i$. This proves our claim and shows that the highest level 
component~$C$ with $u \in X_C$ is uniquely defined. Moreover, we obtain $u \in 
\spc(r_i)$.

Now consider a level-$i'$ component $C'$ with $i' < i$, such that $u \in 
X_{C'}$, and let $C''$ be the unique level-$(i'+1)$ component containing~$C'$. 
We claim that $u \in X_{C''}$. If $u \in C' \subseteq C''$, then $u \in 
X_{C''}$, since~$u \in \spc(r_i)$, $\dist_{C''}(u) = 0 \leq 2r_i$ and $i'+1 
\leq i$. If $u \notin C'$, then $u \in X_{C'}$ implies $u \in I_{C'}$, which 
means that there must be a vertex~$w \in C'$ with~$\dist(u,w) \leq 2 r_i$. But 
then $w \in C''$ and thus $\dist_{C''}(u) \leq 2r_i$. Together with~$u \in 
\spc(r_i)$, this implies~$u \in X_{C''}$, as claimed. Since~$v_{C'}$ is a child 
of $v_{C''}$ in the tree~$D$, it follows inductively that the nodes of $D$ with 
bags containing $u$ induce a subtree of $D$ with root $v_C$, which establishes 
Property~\eqref{item:tw-vertices}.

By \cref{obs_interface} each set of interface points contains at most one hub of 
each level. Since all edges have length at least~$3$, there are no hubs in 
$\spc(r_0)$ on level $0$. This means that each bag of the tree decomposition 
contains at most $1+\lceil\log_2(\alpha)\rceil$ interface points. The bags for 
level-$0$ components contain one additional vertex. Thus the treewidth of $G$ is 
at most $1+\lceil\log_2(\alpha)\rceil$, as claimed.
\qed\end{proof}

An additional property that we will exploit for our algorithms is the 
following. 
A \emph{$(\mu,\delta)$-net} $N\subseteq V$ is a subset of vertices such that 
(a)~the distance between any two distinct \emph{net points} $u,w\in N$ is more 
than~$\mu$, and (b)~for every vertex $v\in V$ there is some net point $w\in N$ 
at distance at most~$\delta$. For graphs of highway dimension~$1$ however, we 
can obtain nets with additional favourable properties, as the next lemma shows.

\begin{lem}\label{lem:net}
For any graph $G$ of highway dimension $1$ and any $r>0$, there is an 
$(r,3r)$-net such that every connected component of $G_{\leq r}$ contains 
exactly one net point. Moreover this net can be computed in polynomial time.
\end{lem}

\begin{proof}
We first derive an upper bound of $3r$ for the diameter of any connected 
component of $G_{\leq r}$. \cref{lem:hd1-G2r} implies that a connected component 
$C$ contains at most one hub $x$ of $\spc(r/2)$. By definition, any shortest 
path in $C$ of length in $(r/2,r]$ must pass through $x$. We also know that 
every edge of $C$ has length at most $r$. Consequently, every edge in $C$ not 
incident to~$x$ must have length at most $r/2$, since each edge constitutes a 
shortest path between its endpoints. This implies that any shortest path in $C$ 
that is not hit by $x$ must have length at most~$r/2$: if $C$ contains a 
shortest path $P$ with length more than $r/2$ not containing~$x$ we could 
repeatedly remove edges of length at most $r/2$ from $P$ until we obtain a 
shortest path of length in $(r/2,r]$ not hit by~$x$, a contradiction. 
Now consider a shortest path $P$ in $G$ of length more than $r/2$ from some 
vertex $v \in C$ to $x$ (note that this path may not be entirely contained in 
$C$). Let $\{u,w\}$ be the unique edge of $P$ such that $\dist(v,u) \leq r/2$ 
and $\dist(v,w) > r/2$.
If the length of the edge $\{u,w\}$ is at most $r/2$ then $\dist(v,w) \leq r$, 
and thus $w=x$, since the part of the path from $v$ to $w$ is a shortest path 
of length in $(r/2,r]$ and thus needs to pass through~$x$. Otherwise the length 
of the edge $\{u,w\}$ is in the interval $(r/2,r]$, which again implies $w=x$, 
since the edge must contain~$x$. In either case, $\dist(v,x) \leq 3r/2$. This 
implies that every vertex in $C$ is at distance at most $3r/2$ from $x$, and 
thus the diameter of $C$ is at most~$3r$.

To compute the $(r,3r)$-net, we greedily pick an arbitrary vertex of each 
connected component of~$G_{\leq r}$. As the distances between components of 
$G_{\leq r}$ is greater than $r$, and every vertex lies in some component 
containing a net point, we get the desired distance bounds. Clearly this net 
can be computed in polynomial time.
\qed\end{proof}

\section{Approximation schemes}
\label{sec:scheme}

In general the aspect ratio of a graph may be exponential in the input size. A 
key ingredient of our algorithms is to reduce the aspect ratio~$\alpha$ of the 
input graph~$G=(V,E)$ to a polynomial. For \pname{STP} and TSP, standard 
techniques can be used to reduce the aspect ratio to $O(n/\eps)$ when aiming 
for a $(1+\eps)$\hy{}approximation. This was for instance also used 
in~\cite{feldmann2018} for low highway dimension graphs, but here we need to 
take special care not to destroy the structural properties given by 
\cref{lem:tw} in this process. In particular, we need to reduce the aspect ratio 
and maintain the fact that the treewidth is bounded.

Therefore, we reduce the aspect ratio of our graphs by the following 
preprocessing. Both metric \pname{TSP} and \pname{STP} admit constant factor 
approximations in polynomial time using well-known 
algorithms~\cite{byrka2010,Chr76TSP}. We first compute a solution of cost~$c$ 
using a $\beta$-approximation algorithm for the problem at hand (\pname{TSP} or 
\pname{STP}). For \pname{TSP}, the diameter of the graph $G$ clearly is at 
most~$c/2$. 
For \pname{STP} we remove every vertex of $V$ that is at distance more than $c$ 
from any terminal, since such a vertex cannot be part of the optimum solution.  
After having removed all such vertices in this way, we obtain a graph $G$ of 
diameter at most~$3c$. Thus, in the following, we may assume that our graph $G$ 
has diameter at most $3c$. We then set $r=\frac{\eps c}{3n}$ in \cref{lem:net} 
to obtain a $(\frac{\eps c}{3n},\frac{\eps c}{n})$-net $N\subseteq V$. As a 
consequence the metric induced by~$N$ (with distances of $G$) has aspect ratio 
at most $\frac{3c}{\eps c/(3n)}=O(n/\eps)$, since the minimum distance between 
any two net points of $N$ is at least $\frac{\eps c}{3n}$ and the maximum 
distance is at most $3c$.
We will exploit this property in the following.

By \cref{lem:net}, each connected component of $G_{\leq\frac{\eps c}{3n}}$ 
contains exactly one net point of $N$. 
Let $\eta\colon V\mapsto N$ map each vertex of $G$ to the unique net point in 
the same connected component of $G_{\leq\frac{\eps c}{3n}}$. 
We define a new graph~$G'$ with vertex set $N \subseteq V$ and edge 
set~$\{\{\eta(u),\eta(v)\}: \{u,v\} \in E \land \eta(u)\neq\eta(v)\}$. 
The length of each edge $\{w,w'\}$ of $G'$ is the shortest path distance between 
$w$ and $w'$ in $G$. This new graph~$G'$ may not have bounded highway dimension, 
but we claim that it has treewidth $O(\log(n/\eps))$. 

\begin{lem}\label{lem:tw2}
If~$G$ has highway dimension~$1$, the graph $G'$ with vertex set $N$ has 
treewidth $O(\log(n/\eps))$. Moreover, a tree decomposition for $G'$ of width 
$O(\log(n/\eps))$ can be computed in polynomial time.
\end{lem}
\begin{proof}
We construct a tree decomposition~$D'$ of~$G'$ as follows. Following
\cref{lem:tw} we can compute a tree decomposition~$D$ of width at most 
$1+\lceil\log_2(\alpha)\rceil$, where $\alpha$ is the aspect ratio of~$G$: for 
this we need to compute a locally $1$-sparse shortest path cover $\spc(r_i)$ for 
each level $i$, which can be done in polynomial time via an XP 
algorithm~\cite{feldmann2018} if the highway dimension is~$1$. We then find the 
level-$i$ components and their interface points, from which the tree 
decomposition~$D$ and its bags can be constructed. Since there are 
$O(\log\alpha)$ levels and $\alpha$ is at most exponential in the input size 
(which includes the encoding length of the edge weights), we can compute $D$ 
in polynomial time. 

We construct~$D'$ from~$D$ by replacing every bag~$X$ of~$D$ by a new bag 
$X'=\{\eta(v) : v\in X\}$ containing the net points for the vertices in~$X$. It 
is not hard to see that Properties~\eqref{item:tw-union} 
and~\eqref{item:tw-edges} 
of \cref{dfn:treewidth} are fulfilled by $D'$, since they are true for~$D$. For 
Property~\eqref{item:tw-vertices}, note that for any edge $\{u,v\}$ of $G$, the 
set of all bags of $D$ that contain $u$ or $v$ form a connected subtree of $D$. 
This is because the bags containing $u$ form a connected subtree 
(Property~\eqref{item:tw-vertices}), the same is true for $v$, and both these 
subtrees share at least one node labelled by a bag containing the edge $\{u,v\}$ 
(Property~\eqref{item:tw-edges}). Consequently, the set of all bags containing 
vertices of any connected subgraph of~$G$ form a connected subtree. In 
particular, for any connected component~$A$ of~$G_{\leq\frac{\eps c}{3n}}$, the 
set of bags of~$D$ containing at least one vertex of~$A$ form a connected 
subtree. This implies Property~\eqref{item:tw-vertices} for $D'$. Thus, $D'$ is 
indeed a tree decomposition of $G'$ according to \cref{dfn:treewidth}. Note that 
$D'$ can be computed in polynomial time.

To bound the width of $D'$, recall that a bag $X$ of the tree decomposition $D$ 
of $G$ contains the interface points $I_C$ of a level-$i$ component~$C$, in 
addition to one more vertex of $C$ on the lowest level~$i=0$. Each interface 
point is a hub from $\spc(r_j)$ at some level $j\geq i$ and is at distance at 
most $2r_j$ from $C$. In particular, if $2r_i\leq\frac{\eps c}{3n}$ then $C$ is 
a component of $G_{\leq 2r_i}\subseteq G_{\leq\frac{\eps c}{3n}}$, and all hubs 
of $I_C\cap\spc(r_j)$ for which $2r_j\leq\frac{\eps c}{3n}$ lie in the same 
connected component $A$ of~$G_{\leq\frac{\eps c}{3n}}$ as $C$. These hubs are 
therefore all mapped to the same net point $w$ in $A$ by $\eta$. In addition to 
$w$, the bag $X'=\{\eta(v) : v\in X\}$ resulting from $X$ and~$\eta$ contains at 
most one vertex for every level $j$ such that $2r_j>\frac{\eps c}{3n}$. As 
$r_j=2^j$, this condition is equivalent to $j>\log_2(\frac{\eps c}{3n})-1$. As 
there are $1+\lceil\log_2(\alpha)\rceil$ levels in total, there are 
$O(\log(\frac{\alpha n}{\eps c}))$ hubs in $X'$. This bound is obviously also 
valid in case $2r_i>\frac{\eps c}{3n}$. We preprocessed the graph $G$ so that 
its diameter is at most $3c$ and its minimum distance is $3$, which implies an 
aspect ratio $\alpha$ of at most $c$ for $G$. This means that every bag $X'$ 
contains $O(\log(n/\eps))$ vertices, and thus the claimed treewidth bound for 
$G'$ follows.
\qed\end{proof}

We are now ready to prove our main result.

\begin{proof}[of~\cref{thm:main}]
To solve \pname{TSP} or \pname{STP} on $G$ we first use the above reduction to 
obtain $G'$ and its tree decomposition~$D'$, and then compute an optimum 
solution for $G'$. For TSP, $G'$ is already a valid input instance, but for 
\pname{STP} we need to define a terminal set, which simply is $R'=\{\eta(v)\mid 
v\in R\}$ if $R$ is the terminal set of~$G$. \citet{bodlaender2013deterministic} 
proved that for both \pname{TSP} and \pname{STP} there are deterministic 
algorithms to solve these problems exactly in time $2^{O(t)}n$, given a tree 
decomposition of the input graph of width $t$. By \cref{lem:tw2} we can thus 
compute the optimum to $G'$ in time $2^{O(\log(n/\eps))}\cdot 
n=(n/\eps)^{O(1)}$. Afterwards, we convert the solution for $G'$ back to a 
solution for $G$, as follows. 

For \pname{TSP} we may greedily add vertices of $V$ to the tour on $N$ by 
connecting every vertex $v\in V$ to the net point $\eta(v)$. As the vertices $N$ 
of $G'$ form a $(\frac{\eps c}{3n},\frac{\eps c}{n})$-net of~$V$, this incurs an 
additional cost of at most $2\frac{\eps c}{n}$ per vertex, which sums up to at 
most $2\eps c$. Let $\OPT$ and $\OPT'$ denote the costs of the optimum tours in 
$G$ and $G'$, respectively. We know that $c \leq \beta\cdot\OPT$, since we used 
a $\beta$-approximation algorithm to compute~$c$. Furthermore, the optimum tour 
in $G$ can be converted to a tour in $G'$ of cost at most $\OPT$ by 
short-cutting, due to the triangle inequality. Thus $\OPT'\leq \OPT$, which 
means that the cost of the computed tour in $G$ is at most $\OPT'+2\eps c\leq 
(1+2\beta\eps) \OPT$.

Similarly, for \pname{STP} we may greedily connect a terminal~$v$ of~$G$ to the 
terminal $\eta(v)$ of~$G'$ in the computed Steiner tree in $G'$. This adds an 
additional cost of at most $\frac{\eps c}{n}$, which sums up to at most $\eps 
c$. Let now $\OPT$ and $\OPT'$ be the costs of the optimum Steiner trees in $G$ 
and $G'$, respectively. We may convert a Steiner tree $T$ in $G$ into a tree 
$T'$ in $G'$ by using edge $\{\eta(u),\eta(v)\}$ for each edge $\{u,v\}$ of $T$. 
Note that the resulting tree $T'$ contains all terminals of $G'$, since 
$R'=\{\eta(v)\mid v\in R\}$. As the vertices $N$ of $G'$ form a $(\frac{\eps 
c}{3n},\frac{\eps c}{n})$-net of~$V$, the cost of $T'$ is at most $\OPT+2\eps c$ 
if the cost of $T$ is $\OPT$ (by the same argument as used for the proof of 
\cref{lem:tw2}). As before, we know that $c\leq \beta\cdot\OPT$, and thus the 
cost of the computed Steiner tree in $G$ is at most $\OPT'+\eps c\leq \OPT+3\eps 
c\leq (1+3\beta\eps) \OPT$.

Hence we obtain FPTASs for both \pname{TSP} and STP, which compute 
$(1+\eps)$\hy{}approximations within a runtime that is polynomial in the input 
size and $1/\eps$.
\qed\end{proof}

We prove next that \pname{STP} is $\mathsf{NP}$-hard on graphs of highway 
dimension $1$, which means that the problem is weakly $\mathsf{NP}$-hard for 
these inputs (cf.~\cite{vazirani01book}). Whether \pname{TSP} is 
$\mathsf{NP}$-hard for such small highway dimension remains open, but we prove 
that it is for highway dimension $6$.

\section{Hardness of \pname{Steiner Tree} for highway dimension 1}
\label{sec:hard}

We present a reduction from the $\mathsf{NP}$-hard satisfiability problem 
(SAT)~\cite{GareyJohnson}, in which a 
Boolean formula~$\varphi$ in conjunctive normal form is given, and a satisfying 
assignment of its variables needs to be found.

\begin{proof}[of \cref{thm:STP-hard}]
For a given SAT formula $\varphi$ with $k$ variables and $\ell$ clauses we 
construct a graph $G_\varphi$ as follows 
(cf.~\cref{fig:hardnessSTP_clause}). For each variable $x$ we introduce a 
path $P_x=(t_x,u_x,f_x)$ with two edges of length $1$ each. The vertex $u_x$ is 
a terminal. Additionally we introduce a terminal $v_0$, which we call the 
\emph{root}, and add the edges $\{v_0,t_x\}$ and $\{v_0,f_x\}$ for every 
variable $x$. Every edge incident to $v_0$ has length $11$. For each clause 
$C_i$, where $i\in\{1,\ldots,\ell\}$, we introduce a terminal $v_i$ and add the 
edge $\{v_i,t_x\}$ for each variable $x$ such that $C_i$ contains $x$ as a 
positive literal, and we add the edge $\{v_i,f_x\}$ for each $x$ for which $C_i$ 
contains $x$ as a negative literal. Every edge incident to $v_i$ has 
length~$11^{i+1}$. Note that the edges incident to the root $v_0$ also have 
length $11^{i+1}$ for $i=0$.

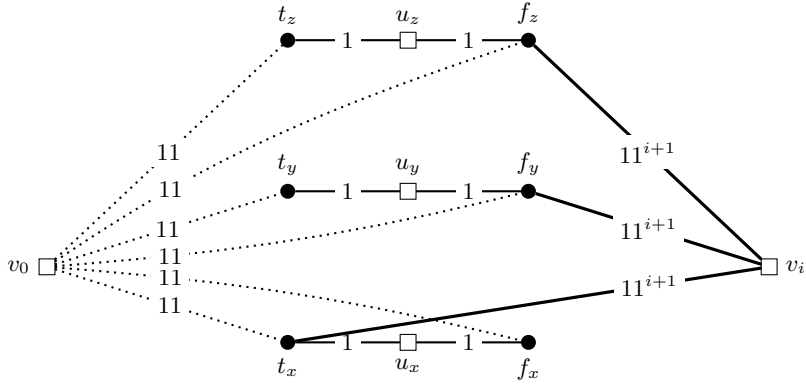
\begin{figure}
\begin{centering}
\begin{tikzpicture}[yscale=0.67,xscale=1.6]
\node[hnode,label=left:{$v_0$}] (v0) at (0,1.5) {};

\node[snode,label=below:{$t_x$}] (tx) at (2,0) {};
\node[hnode,label=below:{$u_x$}] (ux) at (3,0) {};
\node[snode,label=below:{$f_x$}] (fx) at (4,0) {};

\node[snode,label=above:{$t_y$}] (ty) at (2,3) {};
\node[hnode,label=above:{$u_y$}] (uy) at (3,3) {};
\node[snode,label=above:{$f_y$}] (fy) at (4,3) {};

\node[snode,label=above:{$t_z$}] (tz) at (2,6) {};
\node[hnode,label=above:{$u_z$}] (uz) at (3,6) {};
\node[snode,label=above:{$f_z$}] (fz) at (4,6) {};

\node[hnode,label=right:{$v_i$}] (vi) at (6,1.5) {};

\path[draw, arrows={-}, solid, thick] (tx) to [] node[pos=0.5, fill=white, inner sep = 2.5pt] {$1$} (ux);
\path[draw, arrows={-}, solid, thick] (fx) to [] node[pos=0.5, fill=white, inner sep = 2.5pt] {$1$} (ux);
\path[draw, arrows={-}, solid, thick] (ty) to [] node[pos=0.5, fill=white, inner sep = 2.5pt] {$1$} (uy);
\path[draw, arrows={-}, solid, thick] (fy) to [] node[pos=0.5, fill=white, inner sep = 2.5pt] {$1$} (uy);
\path[draw, arrows={-}, solid, thick] (tz) to [] node[pos=0.5, fill=white, inner sep = 2.5pt] {$1$} (uz);
\path[draw, arrows={-}, solid, thick] (fz) to [] node[pos=0.5, fill=white, inner sep = 2.5pt] {$1$} (uz);

\path[draw, arrows={-}, solid, very thick] (vi) to  node[pos=0.24, fill=white, inner sep = 2.5pt] {$11^{i+1}$} (tx);
\path[draw, arrows={-}, solid, very thick] (vi) to node[pos=0.5, fill=white, inner sep = 2.5pt] {$11^{i+1}$} (fy);
\path[draw, arrows={-}, solid, very thick] (vi) to node[pos=0.5, fill=white, inner sep = 2.5pt] {$11^{i+1}$} (fz);

\path[draw, arrows={-}, solid, dotted, thick] (v0) to node[pos=0.51, fill=white, inner sep = 2.5pt] {$11$} (tx);
\path[draw, arrows={-}, solid, dotted, thick,bend left=10] (v0) to node[pos=0.22, fill=white, inner sep = 2.5pt] {$11$} (fx);
\path[draw, arrows={-}, solid, dotted, thick] (v0) to node[pos=0.5, fill=white, inner sep = 2.5pt] {$11$} (ty);
\path[draw, arrows={-}, solid, dotted, thick, bend right = 10] (v0) to [] node[pos=0.22, fill=white, inner sep = 2.5pt,bend left=10] {$11$} (fy);
\path[draw, arrows={-}, solid, dotted, thick] (v0) to node[pos=0.5, fill=white, inner sep = 2.5pt] {$11$} (tz);
\path[draw, arrows={-}, solid, dotted, thick, bend left=10] (v0) to node[pos=0.28, fill=white, inner sep = 2.5pt] {$11$} (fz);

\end{tikzpicture}
\par\end{centering}
\caption{Illustration of the part of the construction involving 
vertices~$x,y,z$ and a clause~$C_i = (x \vee \bar{y} \vee \bar{z})$. Terminals 
are marked as boxes. \label{fig:hardnessSTP_clause}}
\end{figure}

\begin{lem}\label{lem:STP-hd1}
The constructed graph $G_\varphi$ has highway dimension $1$.
\end{lem}
\begin{proof}
Fix a scale $r>0$. If $r\leq 5$ then the shortest path cover $\spc(r)$ only 
needs to hit shortest paths of length at most $2r\leq 10$. Since all edges 
incident to terminals $v_j$ with $j \in \{0,\dots,\ell\}$ have length at 
least~$11$, any such path contains only edges of paths~$P_x$. Thus it suffices 
to include all vertices $u_x$ in~$\spc(r)$. A ball $B_w(2r)$ of radius $2r\leq 
10$ can also only contain some subset of vertices of a single path $P_x$, or a 
single vertex $v_j$. In the former case the ball contains at most the vertex 
$u_x\in\spc(r)$, and in the latter none of $\spc(r)$.

If $r> 5$, let $i=\lfloor \log_{11}(r/5)\rfloor\geq 0$ and $\spc(r)=\{v_i\}$. 
Since there is only one hub, this shortest path cover is locally $1$-sparse. 
Note that any edge incident to a vertex $v_j$ with $j\geq i+1$ has length at 
least $11^{i+2}\geq 11r/5> 2r$. Also, all paths that do not use any $v_j$ with 
$j\geq i$ have length at most $2+\sum_{j=0}^{i-1} (2\cdot 11^{j+1}+2)$, since 
such a path can contain at most two edges incident to a vertex $v_j$ with 
$j\leq i-1$ and the paths $P_x$ of length $2$ are connected only through edges 
incident to vertices~$v_j$. The length of such a path is thus shorter than
\[
2+2\frac{11^{i+1}}{11-1}+2i\leq 3\cdot 11^{i}+2\cdot 11^{i} \leq 5\cdot 11^{i} 
\leq r,
\]
where the first inequality holds since $i+1\leq 11^{i}$ whenever $i\geq 0$. 
Hence the only paths that need to be hit by hubs on scale $r$ are those passing 
through $v_i$, which is a hub of $\spc(r)$.
\qed\end{proof}

To finish the reduction, we claim that there is a satisfying assignment for 
$\varphi$ if and only if there is a Steiner tree $T$ for $G_\varphi$ with cost 
at most $12k+\sum_{i=1}^{\ell} 11^{i+1}$. If there is a satisfying assignment 
for $\varphi$, then the tree $T$ contains the edges $\{u_x,t_x\}$ and 
$\{v_0,t_x\}$ for variables $x$ that are set to true, and the edges 
$\{u_x,f_x\}$ and $\{v_0,f_x\}$ for variables $x$ that are set to false. This 
connects every terminal $u_x$ with the root $v_0$, and the cost of these edges 
is $12k$. For every terminal $v_i$ where $i\geq 1$ we can now add the edge 
$\{v_i,s_x\}$ for $s_x\in\{t_x,f_x\}$ that corresponds to a literal of~$C_i$ 
that is true in the satisfying assignment. Since this Steiner vertex~$s_x$ is 
connected to the root $v_0$, we obtain a Steiner tree $T$. The latter edges add 
another $\sum_{i=1}^{\ell} 11^{i+1}$ to the solution cost, and thus the total 
cost is as claimed.

Conversely, consider a minimum cost Steiner tree $T$ in $G_\varphi$. Note that 
for any terminal $u_x$ the tree must contain an incident edge of cost $1$, while 
for any terminal $v_i$ with $i\geq 1$ the tree must contain an incident edge of 
cost $11^{i+1}$. This adds up to a cost of $k+\sum_{i=1}^\ell 11^{i+1}$. Assume 
that there is some variable $x$ such that $T$ contains neither $\{v_0,t_x\}$ 
nor~$\{v_0,f_x\}$. This means that in $T$ the terminal $u_x$ is connected to the 
root $v_0$ through an edge $\{v_i,s_x\}$ for $s_x\in\{t_x,f_x\}$ and some $i\geq 
1$. The edge $v_0s_x$ forms a fundamental cycle with the tree~$T$, which however 
has a shorter length of $11$ compared to the edge $\{v_i,s_x\}$, which has 
length~$11^{i+1}$. Thus removing $\{v_0,s_x\}$ and adding $\{v_i,s_x\}$ instead, 
would yield a cheaper Steiner tree. As this would contradict that $T$ has 
minimum cost, $T$ contains at least one of the edges $\{v_0,t_x\}$ and 
$\{v_0,f_x\}$ for every variable $x$. This adds another $11k$ to the cost, so 
that $T$ costs at least $12k+\sum_{i=1}^\ell 11^{i+1}$. 

If we assume that $12k+\sum_{i=1}^\ell 11^{i+1}$ is also an upper bound on the 
cost of $T$, by the above observations the tree $T$ contains exactly one edge 
incident to every terminal $u_x$ and $v_i$ for $i\geq 1$, and exactly $k$ edges 
incident to $v_0$. Furthermore, for every variable $x$ the latter edges contain 
exactly one of $\{v_0,t_x\}$ and $\{v_0,f_x\}$. Thus $T$ encodes a satisfying 
assignment for $\varphi$, as follows. For every edge $v_0t_x$ we may set $x$ to 
true, and for every edge $\{v_0,f_x\}$ we may set $x$ to false. For every clause 
$C_i$ the corresponding terminal $v_i$ connects through one of the Steiner 
vertices $s_x\in\{t_x,f_x\}$ of a corresponding literal contained in $C_i$. The 
only incident vertices to $s_x$ in $G_\varphi$ are some terminals~$v_j$, the 
terminal $u_x$, and the root $v_0$. As each $v_j$ and also $u_x$ only has one 
incident edge contained in the tree $T$, the tree must contain the edge 
$\{v_0,s_x\}$ so that the root can be reached from $s_x$ in~$T$. Hence $s_x$ 
corresponds to a literal that is true in $C_i$. Using \cref{lem:STP-hd1}, which 
bounds the highway dimension of $G_\varphi$, we obtain \cref{thm:STP-hard}.
\qed\end{proof}

\section{Hardness of \pname{Travelling Salesperson} for highway dimension 6}

We now show hardness of \pname{TSP} for graphs of bounded highway dimension.
We first introduce a simple lemma that will allow us to easily bound the 
highway dimension of our construction by adding edges incrementally. In the 
following we denote the highway dimension of a graph $G$ by $\hd(G)$.

\begin{dfn}
A cost $c^{\star} \in \mathbb{R}$ is \emph{safe} w.r.t.~a (multi-)set of costs
$C \subseteq \mathbb{R}$ if $c^{\star}\geq2\sum_{c\in C}c$.
\end{dfn}
\begin{lem}
Let $G=(V,E)$ be a graph, $E'\subseteq{V \choose 2}$, and $G'=(V,E\cup E')$.
If the edges in $E'$ have safe costs w.r.t.~the edge costs of $G$,
then $\hd(G')\leq\max\{\hd(G),|E'|\}$.
\label{lem:hardness_hd_edge_addition}
\end{lem}
\begin{proof}
Let $c^{\star}$ be the smallest cost among the edges of $E'$ and
consider a fixed scale $r\in\mathbb{R}^{+}$. If $r<c^{\star}/2$,
then no path of length $l\in(r,2r]$ in $G'$ contains any of the
edges in $E'$. By the definition of the highway dimension of $G$,
there is a locally $\hd(G)$-sparse shortest path cover $\textsc{spc}(r)$
of $G'$ for this scale. Now if $r\geq c^{\star}/2$, then every path
of length $l\in(r,2r]$ in $G'$ must contain an edge of $E'$, since
the costs of all edges in $E$ sum up to at most $r$. Therefore,
we can find a shortest path cover $\textsc{spc}(r)$ of $G'$ simply
by taking a minimum vertex cover of $E'$, which has size at most
$|E'|$.
\qed\end{proof}

We are now ready to prove hardness.

\begin{proof}[of \cref{thm:TSP-hard}.]
\global\long\def\sat{(\leq\!3,3)\textrm{-}\textsc{Sat}}
We reduce from $\sat$~\cite{GareyJohnson}. To that end, let a $\sat$ formula be 
given,
with variables $x_{1},\dots x_{n}$ and clauses $C_{1},\dots,C_{m}$,
where each literal appears at most twice (and each variable at most
three times). We construct a graph $G$ with edge costs taken from
among the values $a\ll b\ll c_{1}\ll\dots\ll c_{n-1}\ll d\ll e\ll 
f_{1}\ll\dots\ll f_{m}$,
where each cost value can be chosen arbitrarily such that it is safe
with respect to the costs of all cheaper edges. For example, there
will be $2n$ edges of cost $a$, hence we choose $b\geq4an$. Let
$T^{\star}$ be (any) $\textsc{TSP}$ tour in $G$ of minimum cost~$|T^{\star}|$.
We consider $T^{\star}$ to be oriented arbitrarily in one of its
two possible orientations.

For every variable $x_{i}$, we introduce a gadget with four vertices
$v_{i1},v_{i2},v_{i3},v_{i4}$ and the edges $\{v_{i1},v_{i3}\}$
and $\{v_{i2},v_{i4}\}$, both of cost $a$. We further add edges
$\{v_{i1},v_{i2}\},\{v_{i2},v_{i3}\},\{v_{i3},v_{i4}\},\{v_{i4},v_{i1}\}$
of cost $b$ each (cf.~\cref{fig:hardness_vertices}). We
will enforce that $T^{\star}$ uses both edges of cost $a$ in every
variable gadget, and we will interpret the variable as 'true', if
the orientation of these edges along the tour is $(v_{i1},v_{i3})$
and $(v_{i4},v_{i2})$, or $(v_{i3},v_{i1})$ and $(v_{i2},v_{i4})$,
and 'false' otherwise. Let $G^{b}$ be the graph we have constructed
so far. Each of the $n$ components of $G^{b}$ has highway dimension
$2$: For scales $r<a$, we can set 
$\textsc{spc}(r)=\bigcup_{i}\{v_{i1},v_{i2}\}$,
and for scales $r\geq a$, we can set 
$\textsc{spc}(r)=\bigcup_{i}\{v_{i1},v_{i3}\}$.
Hence, $\hd(G^{b})=2$.

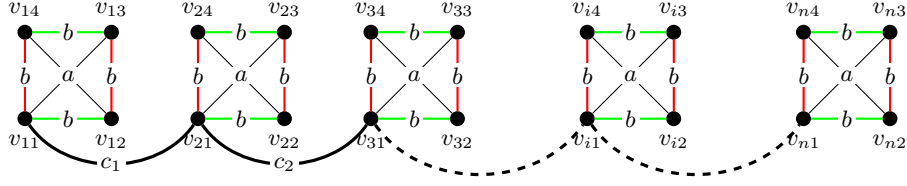
\begin{figure}[h]
\begin{centering}
\begin{tikzpicture}[scale=1.15]
\node[snode,label=above:{$v_{23}$}] (v23) at (3,1) {};
\node[snode,label=below:{$v_{22}$}] (v22) at (3,0) {};
\node[snode,label=below:{$v_{21}$}] (v21) at (2,0) {};
\node[snode,label=above:{$v_{33}$}] (v33) at (5,1) {};
\node[snode,label=above:{$v_{34}$}] (v34) at (4,1) {};
\node[snode,label=below:{$v_{n1}$}] (vn1) at (9,0) {};
\node[snode,label=above:{$v_{24}$}] (v24) at (2,1) {};
\node[snode,label=below:{$v_{12}$}] (v12) at (1,0) {};
\node[snode,label=above:{$v_{13}$}] (v13) at (1,1) {};
\node[snode,label=below:{$v_{11}$}] (v11) at (0,0) {};
\node[snode,label=below:{$v_{31}$}] (v31) at (4,0) {};
\node[snode,label=below:{$v_{32}$}] (v32) at (5,0) {};
\node[snode,label=above:{$v_{14}$}] (v14) at (0,1) {};
\node[snode,label=below:{$v_{i2}$}] (vi2) at (7.5,0) {};
\node[snode,label=above:{$v_{i3}$}] (vi3) at (7.5,1) {};
\node[snode,label=below:{$v_{i1}$}] (vi1) at (6.5,0) {};
\node[snode,label=above:{$v_{n3}$}] (vn3) at (10,1) {};
\node[snode,label=below:{$v_{n2}$}] (vn2) at (10,0) {};
\node[snode,label=above:{$v_{i4}$}] (vi4) at (6.5,1) {};
\node[snode,label=above:{$v_{n4}$}] (vn4) at (9,1) {};
\begin{scope}[>=Stealth]
\path[draw=red, fill=red, arrows={-}, solid, thick] (v14) to node[pos=0.5, fill=white, inner sep = 1.25pt] {$b$} (v11);
\path[draw=green, fill=green, arrows={-}, solid, thick] (v13) to node[pos=0.5, fill=white, inner sep = 1.25pt] {$b$} (v14);
\path[draw=red, fill=red, arrows={-}, solid, thick] (v12) to node[pos=0.5, fill=white, inner sep = 1.25pt] {$b$} (v13);
\path[draw=green, fill=green, arrows={-}, solid, thick] (v11) to node[pos=0.5, fill=white, inner sep = 1.25pt] {$b$} (v12);
\path[draw, arrows={-}, solid, thin] (v12) to node[pos=0.5, fill=white, inner sep = 1.25pt] {$a$} (v14);
\path[draw, arrows={-}, solid, thin] (v11) to node[pos=0.5, fill=white, inner sep = 1.25pt] {$a$} (v13);
\path[draw=red, fill=red, arrows={-}, solid, thick] (v24) to [] node[pos=0.5, fill=white, inner sep = 1.25pt] {$b$} (v21);
\path[draw=green, fill=green, arrows={-}, solid, thick] (v23) to [] node[pos=0.5, fill=white, inner sep = 1.25pt] {$b$} (v24);
\path[draw=red, fill=red, arrows={-}, solid, thick] (v22) to [] node[pos=0.5, fill=white, inner sep = 1.25pt] {$b$} (v23);
\path[draw=green, fill=green, arrows={-}, solid, thick] (v21) to [] node[pos=0.5, fill=white, inner sep = 1.25pt] {$b$} (v22);
\path[draw, arrows={-}, solid, thin] (v22) to [] node[pos=0.5, fill=white, inner sep = 1.25pt] {$a$} (v24);
\path[draw, arrows={-}, solid, thin] (v21) to [] node[pos=0.5, fill=white, inner sep = 1.25pt] {$a$} (v23);
\path[draw=red, fill=red, arrows={-}, solid, thick] (v34) to [] node[pos=0.5, fill=white, inner sep = 1.25pt] {$b$} (v31);
\path[draw=green, fill=green, arrows={-}, solid, thick] (v33) to [] node[pos=0.5, fill=white, inner sep = 1.25pt] {$b$} (v34);
\path[draw=red, fill=red, arrows={-}, solid, thick] (v32) to [] node[pos=0.5, fill=white, inner sep = 1.25pt] {$b$} (v33);
\path[draw=green, fill=green, arrows={-}, solid, thick] (v31) to [] node[pos=0.5, fill=white, inner sep = 1.25pt] {$b$} (v32);
\path[draw, arrows={-}, solid, thin] (v32) to node[pos=0.5, fill=white, inner sep = 1.25pt] {$a$} (v34);
\path[draw, arrows={-}, solid, thin] (v31) to node[pos=0.5, fill=white, inner sep = 1.25pt] {$a$} (v33);
\path[draw=red, fill=red, arrows={-}, solid, thick] (vi4) to node[pos=0.5, fill=white, inner sep = 1.25pt] {$b$} (vi1);
\path[draw=green, fill=green, arrows={-}, solid, thick] (vi3) to node[pos=0.5, fill=white, inner sep = 1.25pt] {$b$} (vi4);
\path[draw=red, fill=red, arrows={-}, solid, thick] (vi2) to [] node[pos=0.5, fill=white, inner sep = 1.25pt] {$b$} (vi3);
\path[draw=green, fill=green, arrows={-}, solid, thick] (vi1) to [] node[pos=0.5, fill=white, inner sep = 1.25pt] {$b$} (vi2);
\path[draw, arrows={-}, solid, thin] (vi2) to [] node[pos=0.5, fill=white, inner sep = 1.25pt] {$a$} (vi4);
\path[draw, arrows={-}, solid, thin] (vi1) to [] node[pos=0.5, fill=white, inner sep = 1.25pt] {$a$} (vi3);
\path[draw=red, fill=red, arrows={-}, solid, thick] (vn4) to [] node[pos=0.5, fill=white, inner sep = 1.25pt] {$b$} (vn1);
\path[draw=green, fill=green, arrows={-}, solid, thick] (vn3) to [] node[pos=0.5, fill=white, inner sep = 1.25pt] {$b$} (vn4);
\path[draw=red, fill=red, arrows={-}, solid, thick] (vn2) to [] node[pos=0.5, fill=white, inner sep = 1.25pt] {$b$} (vn3);
\path[draw=green, fill=green, arrows={-}, solid, thick] (vn1) to [] node[pos=0.5, fill=white, inner sep = 1.25pt] {$b$} (vn2);
\path[draw, arrows={-}, solid, thin] (vn2) to [] node[pos=0.5, fill=white, inner sep = 1.25pt] {$a$} (vn4);
\path[draw, arrows={-}, solid, thin] (vn1) to [] node[pos=0.5, fill=white, inner sep = 1.25pt] {$a$} (vn3);
\path[draw, arrows={-}, dashed, very thick] (vi1) to [bend right=55] (vn1);
\path[draw, arrows={-}, dashed, very thick] (v31) to [bend right=55] (vi1);
\path[draw, arrows={-}, solid, very thick] (v21) to [bend right=55] node[pos=0.5, fill=white, inner sep = 1.25pt] {$c_2$} (v31);
\path[draw, arrows={-}, solid, very thick] (v11) to [bend right=55] node[pos=0.5, fill=white, inner sep = 1.25pt] {$c_1$} (v21);
\end{scope}
\end{tikzpicture}
\par\end{centering}
\caption{Vertex gadgets.\label{fig:hardness_vertices}}
\end{figure}

We connect the variable gadgets by adding edges $\{v_{i1},v_{(i+1)1}\}$
of cost $c_{i}$ for all $i\in\{1,\dots,n-1\}$ in this order. By
definition, each new edge has a safe cost (w.r.t.~all previous edges),
and hence, by \cref{lem:hardness_hd_edge_addition}, the resulting
graph $G^{c}$ has highway dimension $\hd(G^{2})=2$. We will enforce that 
$T^{\star}$ uses each of these edges exactly twice.

For each clause $C_{j}$, we introduce a clause gadget with six vertices
$w_{j1}$, $w'_{j1}$, $w_{j2}$, $w_{j2}'$, $w_{j3}$, $w_{j3}'$ 
(cf.~\cref{fig:hardness_clause}).
We first add three edges 
$\{w_{j1}',w_{j2}\}$, $\{w_{j2}',w_{j3}\}$, $\{w_{j3}',w_{j1}\}$
of cost $d$ each. Since these edges are disconnected, the resulting
graph $G^{d}$ still has $\hd(G^{d})=2$. Now, we add three
edges $\{w_{j1},w_{j1}'\}$, $\{w_{j2},w_{j2}'\}$, $\{w_{j3},w_{j3}'\}$
of cost $e$ each. In the resulting graph $G^{e}$, clauses are still
disconnected. Since we added three edges with safe costs for each
clause, by \cref{lem:hardness_hd_edge_addition}, we have 
$\hd(G^{e})\leq3$.

Finally, we connect each clause gadget for clause $C_{j}$ to the
three variable gadgets corresponding to the variables appearing in
$C_{j}$ (cf.~\cref{fig:hardness_clause}). To this end, we
add six edges per clause, step by step in the order of increasing
clause indices. Let $C_{j}=\lambda_{j1}\vee\lambda_{j2}\vee\lambda_{j3}$
and consider $k\in\{1,2,3\}$. Assume $\lambda_{jk}=x_{i}$, i.e.,
$x_{i}$ appears as a positive literal in $C_{j}$. Let $\delta=0$
if $C_{j}$ is the first clause containing the literal $x_{i}$, i.e.,
$x_{i}\notin C_{j'}$ for $j'<j$, and $\delta=2$ otherwise. We add
the edges $\{w_{jk},v_{i(2+\delta)}\},\{w_{jk}',v_{i(1+\delta)}\}$
of cost $f_{j}$. Now assume $\lambda_{jk}=\bar{x}_{i}$ and let again
$\delta=0$ if $\bar{x}_{i}\notin C_{j'}$ for $j'<j$, and $\delta=2$
otherwise. We add the edges 
$\{w_{jk},v_{i(3-\delta)}\},\{w_{jk}',v_{i(2+\delta)}\}$
of cost $f_{j}$. Since we add six edges of safe costs in each step,
by \cref{lem:hardness_hd_edge_addition}, the final graph $G=G^{f}$
has $\hd(G)\leq6$.

\begin{figure}[h]
\begin{centering}
\begin{center}
\begin{tikzpicture}[scale=1.25]
\node[snode,label=above:{$w_{j3}$}] (w31) at (30:1) {};
\node[snode,label=above right:{$w_{j2}'$}] (w22) at (90:1) {};
\node[snode,label=below left:{$w_{j2}$}] (w21) at (150:1) {};
\node[snode,label=above left:{$w_{j1}'$}] (w12) at (210:1) {};
\node[snode,label=below right:{$w_{j1}$}] (w11) at (270:1) {};
\node[snode,label=below:{$w_{j3}'$}] (w32) at (330:1) {};

\node[name path= Cj, shape=circle, minimum size=5mm] (Cj) at (0,0) {\small$C_j \!=\! \bar{x} \!\vee\! y \!\vee\! \bar{z}$};

\node[snode] (x11) at ($(w12)+(240:2)$) {};
\node[snode] (x12) at ($(w12)+(240:1)$) {};
\node[snode] (x13) at ($(w11)+(240:1)$) {};
\node[snode] (x14) at ($(w11)+(240:2)$) {};
\node[name path= x, shape=circle, minimum size=5mm] (x) at ($(x14)+(0.4,0)$) {$x$};

\node[snode] (y11) at ($(w32)+(0:1)$) {};
\node[snode] (y13) at ($(w31)+(2,0)$) {};
\node[snode] (y12) at ($(w31)+(1,0)$) {};
\node[snode] (y14) at ($(w32)+(2,0)$) {};
\node[name path= y, shape=circle, minimum size=5mm] (y) at ($(y14)+(0.4,0)$) {$y$};

\node[snode] (z14) at ($(w22)+(120:1)$) {};
\node[snode] (z12) at ($(w21)+(120:2)$) {};
\node[snode] (z13) at ($(w22)+(120:2)$) {};
\node[snode] (z11) at ($(w21)+(120:1)$) {};
\node[name path= z, shape=circle, minimum size=5mm] (z) at ($(z13)+(0.4,0)$) {$z$};

\begin{scope}[>=Stealth]
\path[draw, arrows={-}, solid, very thick] (w32) to [] node[pos=0.5, fill=white, inner sep = 1.25pt] {$d$} (w11);
\path[draw, arrows={-}, solid, very thick] (w22) to [] node[pos=0.5, fill=white, inner sep = 1.25pt] {$d$} (w31);
\path[draw, arrows={-}, solid, very thick] (w12) to [] node[pos=0.5, fill=white, inner sep = 1.25pt] {$d$} (w21);
\path[draw, arrows={-}, solid, thick] (w31) to [] node[pos=0.5, fill=white, inner sep = 1.25pt] {$e$} (w32);
\path[draw, arrows={-}, solid, thick] (w21) to [] node[pos=0.5, fill=white, inner sep = 1.25pt] {$e$} (w22);
\path[draw, arrows={-}, solid, thick] (w11) to [] node[pos=0.5, fill=white, inner sep = 1.25pt] {$e$} (w12);
\path[draw=red, fill=red, arrows={-}, solid, thick] (x14) to [] node[pos=0.5, fill=white, inner sep = 1.25pt] {$b$} (x11);
\path[draw=green, fill=green, arrows={-}, solid, thick] (x13) to [] node[pos=0.5, fill=white, inner sep = 1.25pt] {$b$} (x14);
\path[draw=red, fill=red, arrows={-}, solid, thick] (x12) to [] node[pos=0.5, fill=white, inner sep = 1.25pt] {$b$} (x13);
\path[draw=green, fill=green, arrows={-}, solid, thick] (x11) to [] node[pos=0.5, fill=white, inner sep = 1.25pt] {$b$} (x12);
\path[draw, arrows={-}, solid, thin] (x12) to [] node[pos=0.5, fill=white, inner sep = 1.25pt] {$a$} (x14);
\path[draw, arrows={-}, solid, thin] (x11) to [] node[pos=0.5, fill=white, inner sep = 1.25pt] {$a$} (x13);
\path[draw=red, fill=red, arrows={-}, solid, thick] (y14) to [] node[pos=0.5, fill=white, inner sep = 1.25pt] {$b$} (y11);
\path[draw=green, fill=green, arrows={-}, solid, thick] (y13) to [] node[pos=0.5, fill=white, inner sep = 1.25pt] {$b$} (y14);
\path[draw=red, fill=red, arrows={-}, solid, thick] (y12) to [] node[pos=0.5, fill=white, inner sep = 1.25pt] {$b$} (y13);
\path[draw=green, fill=green, arrows={-}, solid, thick] (y11) to [] node[pos=0.5, fill=white, inner sep = 1.25pt] {$b$} (y12);
\path[draw, arrows={-}, solid, thin] (y12) to [] node[pos=0.5, fill=white, inner sep = 1.25pt] {$a$} (y14);
\path[draw, arrows={-}, solid, thin] (y11) to [] node[pos=0.5, fill=white, inner sep = 1.25pt] {$a$} (y13);
\path[draw=red, fill=red, arrows={-}, solid, thick] (z14) to [] node[pos=0.5, fill=white, inner sep = 1.25pt] {$b$} (z11);
\path[draw=green, fill=green, arrows={-}, solid, thick] (z13) to [] node[pos=0.5, fill=white, inner sep = 1.25pt] {$b$} (z14);
\path[draw=red, fill=red, arrows={-}, solid, thick] (z12) to [] node[pos=0.5, fill=white, inner sep = 1.25pt] {$b$} (z13);
\path[draw=green, fill=green, arrows={-}, solid, thick] (z11) to [] node[pos=0.5, fill=white, inner sep = 1.25pt] {$b$} (z12);
\path[draw, arrows={-}, solid, thin] (z12) to [] node[pos=0.5, fill=white, inner sep = 1.25pt] {$a$} (z14);
\path[draw, arrows={-}, solid, thin] (z11) to [] node[pos=0.5, fill=white, inner sep = 1.25pt] {$a$} (z13);
\path[draw, arrows={-}, solid, ultra thick] (w12) to [] node[pos=0.5, fill=white, inner sep = 1.25pt] {$f_j$} (x12);
\path[draw, arrows={-}, solid, ultra thick] (w11) to [] node[pos=0.5, fill=white, inner sep = 1.25pt] {$f_j$} (x13);
\path[draw, arrows={-}, solid, ultra thick] (w22) to [] node[pos=0.5, fill=white, inner sep = 1.25pt] {$f_j$} (z14);
\path[draw, arrows={-}, solid, ultra thick] (w21) to [] node[pos=0.5, fill=white, inner sep = 1.25pt] {$f_j$} (z11);
\path[draw, arrows={-}, solid, ultra thick] (w32) to [] node[pos=0.5, fill=white, inner sep = 1.25pt] {$f_j$} (y11);
\path[draw, arrows={-}, solid, ultra thick] (w31) to [] node[pos=0.5, fill=white, inner sep = 1.25pt] {$f_j$} (y12);
\end{scope}
\end{tikzpicture}
\par\end{center}
\par\end{centering}
\caption{Clause gadget.\label{fig:hardness_clause}}
\end{figure}
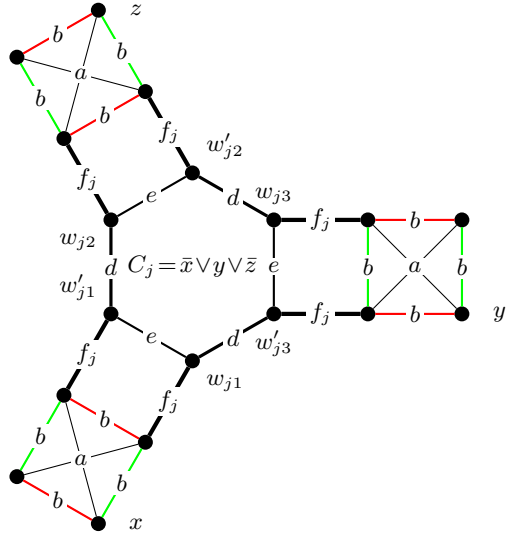

Now let $W=2\sum_{j=1}^{m}f_{j}+2me+3md+2\sum_{i=1}^{n-1}c_{i}+(2n-m)b+2na$.
We claim that $|T^{\star}|\leq W$ if and only if the $\sat$ formula
is satisfiable.

For the first part of the claim, assume the $\sat$ formula is satisfiable,
and, for all $j\in\{1,\dots,m\}$, let $y_{j}$ be a unique variable
that satisfies clause $C_{j}$ in the corresponding assignment. We
describe a tour of cost $W$ by constructing a Eulerian graph consisting
of edges of $G$ (sometimes twice) that connect all vertices with
a total cost of $W$. We start by including the cycle of cost $3e+3d$
within each clause gadget, and each edge between different variable
gadgets twice, for a total cost of $3me+3md+2\sum_{i}^{n-1}c_{i}$.
For every variable $x_{i}$ that is set to true, we include the cycle
$v_{i1},v_{i3},v_{i4},v_{i2},v_{i1}$ of cost $2b+2a$. For every
variable $x_{i}$ that is set to false, we include the cycle 
$v_{i1},v_{i3},v_{i2},v_{i4},v_{i1}$
of cost $2b+2a$. The resulting graph $T'$ is Eulerian, since we
added only cycles, but not yet connected. Its cost is 
$3me+3md+2\sum_{i=1}^{n-1}c_{i}+2nb+2na$. 

Now take a clause $C_{j}$ that is satisfied by variable $x_{i}=y_{j}$.
We add the edges $\{w_{jk},v_{ir}\},\{w_{jk}',v_{i(r+3\mod4)}\}$
that connect the corresponding clause and variable gadgets. Observe
that the edge $\{w_{jk},w_{jk}'\}$ is in $T'$ and so is the edge
$\{v_{ir},v_{i(r+3\mod4)}\}$, because $x_{i}$ satisfies $C_{j}$.
We can thus remove these two edges and obtain a Eulerian graph. This
increases the cost of the graph by $2f_{j}-e-b$. The final graph
$T$ is connected, Eulerian, and has cost $W$ as claimed.

For the second part of the claim, consider any $\textsc{TSP}$ tour
$T$ with $|T|\leq W$. Observe that $3f_{m}>W$, $3f_{m-1}>W-2f_{m}$,
and so on. Since, for all $j\in\{1,\dots,m$\}, the edges of cost
$f_{j}$ form a cut of $G$, we can conclude that $T$ uses exactly
2 edges of cost $f_{j}$ (or one of them twice). Similarly, 
$(2m+1)e>W-\sum_{j=1}^{n-1}f_{j}$,
but $T$ needs to use at least two edges of cost $e$ (or one of them
twice) to connect all vertices of a clause gadget for $C_{j}$ to
the two edges of cost $f_{j}$ that are part of the tour. We can again
conclude that $T$ uses exactly two edges of cost $e$ in each clause
gadget. And again $(3m+1)d>W-\sum_{j=1}^{n-1}f_{j}-2me$, but $T$
needs to use at least three times an edge of cost $d$ to connect
all vertices of a clause gadget, provided that it uses only two edges
of cost $e$. We conclude that $T$ uses exactly three edges of cost
$d$ in each clause gadget. Finally, observe that the only way to
connect all vertices of a clause gadget with two edges of cost $e$
and three edges of cost~$d$ needs that the two edges of cost $f_{j}$
are distinct and connect to the same vertex gadget. We will rely on
this observation in the following, and on the fact that costs are chosen
to be safe with respect to all smaller costs (in particular, $e$ is safe 
with respect to $a,b,c,d$, etc.).

The first implication is that $T$ needs to use every edge between
vertex gadgets twice, since the vertex gadgets are not connected via
clause gadgets. Our analysis so far implies that edges within variable
gadgets that are used in $T$ incur a cost of at most $W'=(2n-m)b+2na$.
Let $k\in\{0,\dots4\}$ be the number of clause gadgets in $T$ connected
to the variable gadget of $x_{i}$. Clearly, $T$ needs to use at
least $4-k$ edges within the variable gadget. The cost within a variable
gadget depending on $k$ is at least $2b+2a$ (if $k=0$), $b+2a$
(if $k=1$), $2a$ (if $k=2$), $a$ (if $k=1$), or 0 (if $k=4$).
Since each variable appears in at most four clauses and each clause
has at most 3 literals, we have $m\leq4n/3<2n$. To obtain a cost
of at most $W'$, we must thus have $k\leq2$ in each variable gadget,
since $b\gg a$. Furthermore, if there are two clause gadgets connected
to a variable gadget, they must connect to disjoint vertices of the
clause to allow for a cost of at most $2a$ in the variable gadget. This
means that the corresponding literals must either both be positive
or both be negative. But if there is an assignment of clauses to variables
such that at most two clauses are assigned to each variable and the
corresponding literal of the assigned clauses must agree, this immediately
yields a satisfying assignment of the $\sat$-formula.
\qed\end{proof}

\section{Conclusions}\label{sec:concl}

We showed that, somewhat surprisingly, graphs of highway dimension $1$ exhibit 
a rich combinatorial structure. On one hand, it was already 
known~\cite{feldmann2018} that these graphs are not minor-closed and thus their 
treewidth is unbounded. Here we additionally showed that \pname{STP} is weakly 
$\mathsf{NP}$-hard on such graphs, further confirming that these graphs have 
non-trivial properties. On the other hand, we proved in \cref{lem:tw} that the 
treewidth of a graph of highway dimension $1$ is logarithmically bounded in the 
aspect ratio $\alpha$. This in turn can be exploited to obtain a very efficient 
FPTAS for both \pname{STP} and \pname{TSP}.

At this point one may wonder whether it is possible to generalize \cref{lem:tw} 
to larger values of the highway dimension. In particular, in~\cite{feldmann2018} 
it was suggested that the treewidth of a graph of highway dimension $h$ might 
be bounded by, say, $O(h\,\textrm{polylog}(\alpha))$. However such a bound is 
highly unlikely in general, since it would have the following consequence for 
the \pname{$k$-Center} problem, for which $k$ vertices (centers) need to be 
selected in a graph such that the maximum distance of any vertex to its closest 
center is minimized. It was shown in~\cite{Feldmann15} that it is $\mathsf{NP}$-hard to 
compute a $(2-\eps)$-approximation for \pname{$k$-Center} on graphs of highway 
dimension $O(\log^2 n)$, for any $\eps>0$. Given such a graph, the same 
preprocessing of \cref{sec:scheme} could be used to derive an analogue of 
\cref{lem:tw2}, i.e., a graph $G'$ of treewidth $O(\textrm{polylog}(n/\eps))$ 
could be computed for the net $N$. Moreover, a $2$-approximation for 
\pname{$k$-Center} can be computed in polynomial time on any 
graph~\cite{hochbaum1986bottleneck}, and if the input has treewidth $t$ a 
$(1+\eps)$-approximation can be computed in $(t/\eps)^{O(t)} n^{O(1)}$ 
time~\cite{katsikarelis2017structural}. Using the same arguments to prove 
\cref{thm:main} for \pname{STP} and \pname{TSP}, it would now be possible to 
compute a $(1+\eps)$-approximation for \pname{$k$-Center} in quasi-polynomial 
time (cf.~\cite{DBLP:conf/swat/FeldmannM18}). That is, we would obtain a QPTAS 
for graphs of highway dimension $O(\log^2 n)$, which is highly unlikely given 
that computing a $(2-\eps)$-approximation is $\mathsf{NP}$-hard on such graphs.

The above argument rules out any bound of $(h\log\alpha)^{O(1)}$ for graphs of 
highway dimension $h$ and aspect ratio $\alpha$, unless $\mathsf{NP}$-hard problems admit 
quasi-polynomial time algorithms. In fact, we conjecture that the 
\pname{$k$-Center} problem is $\mathsf{NP}$-hard to approximate within a factor of $2-\eps$ 
for graphs of constant highway dimension (for some constant larger than $1$). If 
this is true, then the above argument even rules out a treewidth bound of 
$f(h)\,\textrm{polylog}(\alpha)$ for any function~$f$. Thus, in order to answer 
the open problem of~\cite{feldmann2018} and obtain a PTAS for graphs of 
constant highway dimension, a different approach seems to be needed.

\printbibliography

\end{document}